\documentclass[10pt, conference]{IEEEtran}

\usepackage[nospace,nocompress]{cite}

\usepackage{amssymb}
\usepackage{amsfonts}
\usepackage{algorithmic}
\usepackage{algorithm}
\usepackage{listings}
\usepackage{fancyvrb}

\ifCLASSINFOpdf
\usepackage[pdftex]{graphicx}
\graphicspath{{pdf/}{jpeg/}}
\else
\usepackage[dvips]{graphicx}
\graphicspath{{eps/}}
\fi

\usepackage[cmex10]{amsmath}
\usepackage{algorithmic}

\usepackage{array}
\usepackage{mdwmath}
\usepackage{mdwtab}
\usepackage{paralist}
\usepackage{stmaryrd}

\usepackage[caption=false,font=footnotesize]{subfig}
\usepackage{fixltx2e}
\usepackage{stfloats}

\usepackage{url}

\DeclareMathAlphabet{\mathpzc}{OT1}{pzc}{m}{it}

\usepackage{amsthm}
\newtheorem{definition}{Definition}
\newtheorem{lemma}{Lemma}
\newtheorem{theorem}{Theorem}
\newtheorem{corollary}{Corollary}

\begin{document}
\title{A New Enforcement on Declassification with Reachability Analysis}

\author{\IEEEauthorblockN{Cong Sun\qquad Liyong Tang\qquad Zhong
Chen}
\IEEEauthorblockA{Institute of Software, School of Electronics Engineering and Computer Science, Peking University, China\\
Key Laboratory of High Confidence Software Technologies, Ministry of Education, China\\
Key Laboratory of Network and Software Security Assurance, Ministry of Education, China\\
Email: \{suncong,tly,chen\}@infosec.pku.edu.cn}
}


%

\maketitle

\begin{abstract}
Language-based information flow security aims to
decide whether an action-observable program can unintentionally leak
confidential information if it has the authority to access
confidential data. Recent concerns about declassification polices
have provided many choices for practical intended information
release, but more precise enforcement mechanism for these policies
is insufficiently studied. In this paper, we propose a security
property on the where-dimension of declassification and present an
enforcement based on automated verification. The approach
automatically transforms the abstract model with a variant of
self-composition, and checks the reachability of illegal-flow state
of the model after transformation. The self-composition is equipped
with a store-match pattern to reduce the state space and to model
the equivalence of declassified expressions in the premise of
property. The evaluation shows that our approach is more precise
than type-based enforcement.
\end{abstract}

\begin{IEEEkeywords}
information flow security; declassification; pushdown system;
program analysis
\end{IEEEkeywords}

%
\IEEEpeerreviewmaketitle

\section{Introduction}

Information flow security is concerned with finding new techniques
to ensure that the confidential data will not be illegally leaked to
the public observation. The topic is popular at both language
level and operating system level. Language-based techniques have
been pervasively adopted in the study on information flow security.
This is comprehensively surveyed in \cite{Sabelfeld03}.
Noninterference\cite{DBLP:conf/sp/GoguenM82} is commonly known as
the baseline property of information flow security. The
semantic-based definition of
noninterference\cite{DBLP:journals/lisp/SabelfeldS01} on batch-job
model characterizes a security condition specifying that the system
behavior is indistinguishable from a perspective of attacker
regardless of the confidential inputs. Noninterference is criticized
for the restriction that forbids any flow from \emph{high} to
\emph{low}. It will influence the usability of system because the
deliberate release is pervasive in many situations, e.g. password
authentication, online shopping and encryption. Therefore, it is
important to specify more relaxed and practical policies for
real application scenarios and develop precise enforcement
mechanisms for these policies.

The confidentiality aspect of information downgrading, i.e.
declassification\cite{DBLP:conf/sosp/MyersL97}, allows information
release with different intentions along four
dimensions\cite{DBLP:journals/jcs/SabelfeldS09}: \emph{what} is
released, \emph{where} does the release happen, \emph{when} the
information can be released and \emph{who} releases it. The security
policy we propose is on the where-dimension. On this dimension, there
have been several polices, e.g. \emph{intransitive
noninterference}\cite{DBLP:conf/aplas/MantelS04},
\emph{non-disclosure}\cite{DBLP:journals/jcs/MatosB09},
\emph{WHERE}\cite{DBLP:conf/esop/MantelR07}, \emph{flow
locks}\cite{DBLP:conf/esop/BrobergS06}, and \emph{gradual
release}\cite{DBLP:conf/sp/AskarovS07}. Each of them leverages a certain category of
type system to enforce the security policy.

In this work, we first use an approach based on automated
verification to enforce declassification policy on the
where-dimension. As a flow-sensitive and context-sensitive
technique, automated verification has been used as an enforcement to
noninterference on both imperative
languages\cite{DBLP:conf/csfw/BartheDR04,DBLP:conf/sas/TerauchiA05}
and object-oriented
languages\cite{DBLP:conf/esorics/Naumann06,suncong2}. In these works
declassification is only discussed in
\cite{DBLP:conf/sas/TerauchiA05}, where the specific property
\emph{relaxed noninterference}\cite{DBLP:conf/popl/LiZ05} is mostly
on the what-dimension.

The approaches based on automated verification usually rely on some
form of self-composition\cite{DBLP:conf/csfw/BartheDR04} that
composes the program model with a variable-renamed copy to reduce
the security property on original model to a safety property on the
model after transformation. In our previous work\cite{suncong2}, we
have developed a framework that uses reachability analysis to ease
the specification of temporal logic formula or the manual assertion
encoding partial correctness judgement. The self-composition doubles
the size of memory store and largely increases the state space of
model. When the I/O channels are considered, this effect becomes
more serious since each store of channel is modeled explicitly. On
the other hand, the security property often requires the equivalence
of declassified expressions to be satisfied. Therefore in our
enforcement we propose a store-match pattern to 1. avoid duplicating
the output channels, and 2. facilitate the self-composition by
modeling the equivalence of declassified expressions in the premise
of security property. We also evaluated the similarity of the
properties and the preciseness of our enforcement mechanism compared
with type system.

The main contributions of the paper include: (i) We propose a
more relaxed security property enforceable with automated
verification on the where-dimension; (ii) We give a flow-sensitive
and context-sensitive enforcement based on reachability analysis of
pushdown system. We show the mechanism is more precise than
type-based approaches; (iii) We propose a store-match pattern that
can be in common use for automated verifications to reduce the state
space of model and the cost of security analysis.

The rest of the paper is organized as follows. In
Section~\ref{sec:language}, we introduce the language model and the
baseline property. In Section~\ref{sec:property}, we define the
where-security and prove the compliance of property with the prudent
principles. Section~\ref{sec:enforcement} describes the enforcement
mechanism. We show the evaluation in Section~\ref{sec:evaluation}
and conclude in Section~\ref{sec:conclusion}.

\section{Program Model and Baseline Property}\label{sec:language}

We use a sequential imperative language with I/O channels as the
presentation language to illustrate our approach. The syntax is
listed in Fig.\ref{fig:syntax}. The language is deterministic. The
primitive $\mathpzc{declass}$ stands for declassification that
downgrades the confidential data of expression $e$ to be assigned to
variable $x$ with a lower security domain. Here $x$ can be
considered as a low-level sink of data observable to the attacker.
$\mathcal{I}$ and $\mathcal{O}$ are respectively the set of input
and output channels. They are formally defined as a mapping from
each channel identifier $i$ to a linear list, e.g. $\mathcal{I}_i$
resp. $\mathcal{O}_i$. The command $\textit{input}(x,\mathcal{I}_i)$
indicates that the input from $\mathcal{I}_i$ is assigned to $x$,
and the command $\textit{output}(e,\mathcal{O}_i)$ stores the value
of expression $e$ into the correct position of $\mathcal{O}_i$.

\begin{figure}[!t]\footnotesize\centering
\begin{align*}
e ::= & v\mid x\mid e\oplus e'\\
C ::= & \textbf{skip}\mid x:=e \mid x:=\mathpzc{declass}(e) \mid
\textbf{if }e\textbf{ then }C\textbf{ else }C' \mid\\
 &\textbf{while }e\textbf{ do }C \mid C;C' \mid \textit{input}(x,\mathcal{I}_i) \mid \textit{output}(e,\mathcal{O}_i)
\end{align*}
\caption{Program Syntax}\label{fig:syntax}
\end{figure}

\begin{figure}[!t]\footnotesize\centering
\begin{equation*}
\frac{\displaystyle }{\displaystyle
(\mu,\mathcal{I},\mathcal{O},p,q,\textbf{skip};C)\rightarrow
(\mu,\mathcal{I},\mathcal{O},p,q,C)}
\end{equation*}
\begin{equation*}
\frac{\displaystyle \mu(e)=v}{\displaystyle
(\mu,\mathcal{I},\mathcal{O},p,q,x:=e;C)\rightarrow (\mu[x\mapsto
v],\mathcal{I},\mathcal{O},p,q,C)}
\end{equation*}
\begin{equation*}
\frac{\displaystyle \mu(e)=b}{\displaystyle
(\mu,\mathcal{I},\mathcal{O},p,q,\textbf{if }e\textbf{ then
}C_\textbf{true}\textbf{ else
}C_\textbf{false})\rightarrow(\mu,\mathcal{I},\mathcal{O},p,q,C_b)}
\end{equation*}
\begin{equation*}
\frac{\displaystyle \mu(e)=\textbf{true}}{\displaystyle
(\mu,\mathcal{I},\mathcal{O},p,q,\textbf{while }e\textbf{ do
}C)\rightarrow(\mu,\mathcal{I},\mathcal{O},p,q,C;\textbf{while
}e\textbf{ do }C)}
\end{equation*}
\begin{equation*}
\frac{\displaystyle \mu(e)=\textbf{false}}{\displaystyle
(\mu,\mathcal{I},\mathcal{O},p,q,\textbf{while }e\textbf{ do
}C)\rightarrow(\mu,\mathcal{I},\mathcal{O},p,q,\textbf{skip})}
\end{equation*}
\begin{equation*}
\frac{\displaystyle
(\mu,\mathcal{I},\mathcal{O},p,q,C_1)\rightarrow(\mu',\mathcal{I}',\mathcal{O}',p',q',C_1')}{\displaystyle
(\mu,\mathcal{I},\mathcal{O},p,q,C_1;C_2)\rightarrow(\mu',\mathcal{I}',\mathcal{O}',p',q',C_1';C_2)}
\end{equation*}
\begin{equation*}
\frac{\displaystyle \mathcal{I}_i[p_i]=v \qquad
p_i'=p_i+1}{\displaystyle
(\mu,\mathcal{I},\mathcal{O},p,q,\textit{input}(x,\mathcal{I}_i);C)\rightarrow(\mu[x\mapsto
v],\mathcal{I},\mathcal{O},p',q,C)}
\end{equation*}
\begin{equation*}
\frac{\displaystyle \mu(e)=\mathcal{O}_i'[q_i] \qquad
q_i'=q_i+1}{\displaystyle
(\mu,\mathcal{I},\mathcal{O},p,q,\textit{output}(e,\mathcal{O}_i);C)\rightarrow(\mu,\mathcal{I},\mathcal{O}',p,q',C)}
\end{equation*}
\begin{equation*}
\frac{\displaystyle \mu(e)=v \qquad
\sigma(e)\preceq\sigma(x)}{\displaystyle
(\mu,\mathcal{I},\mathcal{O},p,q,x:=\mathpzc{declass}(e);C)\rightarrow(\mu[x\mapsto
v],\mathcal{I},\mathcal{O},p,q,C)}
\end{equation*}
\begin{equation*}
\frac{\displaystyle \mu(e)=v \qquad \sigma(x)\prec\sigma(e) \qquad
\sigma(e)\rightsquigarrow \sigma(x)}{\displaystyle
(\mu,\mathcal{I},\mathcal{O},p,q,x:=\mathpzc{declass}(e);C)\rightarrow_d
(\mu[x\mapsto v],\mathcal{I},\mathcal{O},p,q,C)}
\end{equation*}
\caption{Operational Semantics}\label{fig:semantics}
\end{figure}

The computation is modeled by the small-step operational semantics
in Fig.\ref{fig:semantics}. The inductive rules are defined over
configurations of the form $(\mu,\mathcal{I},\mathcal{O},p,q,C)$.
$\mu:\textit{Var}\mapsto\mathbb{N}$ is a memory store mapping
variables to values and $C$ is the command to be executed. $p$ and
$q$ are set of indices. $p_i$ denotes the index of next element to
be input from $\mathcal{I}_i$, and $q_i$ is the index of location of
$\mathcal{O}_i$ where the next output value will be stored. The
elements in $p$ and $q$ are explicitly increased by the computation
of inputs and outputs.

The security policy is a tuple
$(\mathcal{D},\preceq,\rightsquigarrow,\sigma)$ where
$(\mathcal{D},\preceq)$ is a finite security lattice on security
domains and $\rightsquigarrow$ is an exceptional downgrading
relation of security domains
($\rightsquigarrow\cap\preceq=\emptyset$) statically gathered from
the program. Let
$\sigma:\textit{Var}\cup\mathcal{I}\cup\mathcal{O}\mapsto\mathcal{D}$
be the mapping from I/O channels and variables to security domains,
and let $\sigma(e)\equiv\bigsqcup_{x\in e}\sigma(x)$ be the least
upper bound of the security domains of variables contained in $e$.
When command $x:=\mathpzc{declass}(e)$ in program has
$\sigma(x)\prec\sigma(e)$, the $\mathpzc{declass}$ operation
performs a real downgrading from some variable in $e$ and only then
an element $(\sigma(e),\sigma(x))$ is contained in the relation
$\rightsquigarrow$, otherwise the operation is identical to an
ordinary assignment. We label the transition of declassification
with $\rightarrow_d$ in Fig.\ref{fig:semantics}. The security policy
is different from the MLS policy with exceptions proposed in
\cite{DBLP:conf/aplas/MantelS04,DBLP:conf/esop/MantelR07,DBLP:conf/ifip1-7/LuxM08},
where the set of exceptional relations $\rightsquigarrow$ is
independent to the declassification operations. In our policy the
exceptions are gathered from the $\mathpzc{declass}$ commands. Our
treatment is reasonable since developer should have right to decide
the exception when they use the primitive $\mathpzc{declass}$
explicitly. This is also supported in other work,
e.g.\cite{DBLP:conf/isss2/SabelfeldM03}.

We specify noninterference with the semantic-based
PER-model\cite{DBLP:journals/lisp/SabelfeldS01}. Intuitively
speaking, it specifies a relation between states of any two
correlative runs of program, which is variation in the confidential
initial state cannot cause variation in the public final state. In
another word, the runs starting from indistinguishable initial
states derive indistinguishable final states as well. For the
language with I/Os, the indistinguishability relation on memory
stores and I/O channels with respect to certain security domain
$\ell$ is defined as below.
\begin{definition}[$\ell$-indistinguishability]Memory store $\mu_i$
and $\mu_j$ are indistinguishable on $\ell(\ell\in\mathcal{D})$,
denoted by $\mu_i\sim_\ell\mu_j$, iff $\forall x\in
\textit{Var}.\sigma(x)\preceq\ell\Rightarrow\mu_i(x)=\mu_j(x)$. For
input channel $\mathcal{I}_i$ and $\mathcal{I}_j$,
$\mathcal{I}_i\sim_\ell\mathcal{I}_j$ iff
$(\sigma(\mathcal{I}_i)=\sigma(\mathcal{I}_j)\preceq\ell)\wedge(p_i=p_j\wedge\forall
0\leq k<p_i.\mathcal{I}_i[k]=\mathcal{I}_j[k])$. Similarly, for
output channel $\mathcal{O}_i$ and $\mathcal{O}_j$,
$\mathcal{O}_i\sim_\ell\mathcal{O}_j$ iff
$(\sigma(\mathcal{O}_i)=\sigma(\mathcal{O}_j)\preceq\ell)\wedge(q_i=q_j\wedge\forall
0\leq k<q_i.\mathcal{O}_i[k]=\mathcal{O}_j[k])$.
\end{definition}
For the two observable channels with same security domain, the
indistinguishable linear lists should have the same length and
identical content. Let $\mathcal{I}^\ell$ be the set of input
channels with security domain $\ell'(\ell'\preceq\ell)$. If the set
$\mathcal{I}$ and $\mathcal{I}'$ have the same domain, e.g. as the
inputs of the same program, we can use
$\mathcal{I}\sim_\ell\mathcal{I}'$ to express $\forall
i.\mathcal{I}_i\in\mathcal{I}^\ell\Rightarrow
\mathcal{I}_i\sim_\ell\mathcal{I}_i'$. The noninterference
formalized here takes into consideration the I/O channels and is
therefore different from what for batch-job model\cite{Sabelfeld03}.
It is given as follows.
\begin{definition}[Noninterference]\label{def:NI}
Program $P$ satisfies noninterference w.r.t. security domain
$\ell_0$, iff $\forall\ell\preceq\ell_0$, we have\\
$\left(\begin{array}{l}
\forall\mathcal{I},\mu,\mathcal{I}',\mu',\mathcal{O}_f,\mu_f.(\mu,\mathcal{I},\mathcal{O},p,q,P)\rightarrow^*\\
(\mu_f,\mathcal{I},\mathcal{O}_f,p_f,q_f,\textbf{skip})\wedge
\mathcal{I}\sim_\ell\mathcal{I}'\wedge\mu\sim_\ell\mu'
\end{array}\right)\Rightarrow\\
\left(\begin{array}{l}
\exists\mathcal{O}_f',\mu_f'.(\mu',\mathcal{I}',\mathcal{O}',p',q',P)\rightarrow^*\\
(\mu_f',\mathcal{I}',\mathcal{O}_f',p_f',q_f',\textbf{skip})\wedge\mathcal{O}_f\sim_\ell\mathcal{O}_f'\wedge\mu_f\sim_\ell\mu_f'
\end{array}\right)$.
\end{definition}
In this definition, the noninterference property is related to a
security domain $\ell_0$. The content of channels with security
domain $\ell'(\ell'\succ\ell_0)$ is unobservable and irrelevant to
the property. A more specific way to define noninterference is to
require $\ell_0=\bigsqcup\mathcal{D}$. That means the proposition in
Definition \ref{def:NI} has to be satisfied for each security domain
in $\mathcal{D}$. We use this definition in the following. Our
definition adopts a manner to consider the indistinguishability of
the initial and final states but not to characterize the relation in
each computation step as did by the bisimulation-based
approach\cite{DBLP:conf/csfw/SabelfeldS00}. Another use of the
security domain of variables is to specify where a valid
declassification occurs. This will be discussed below.

\section{Where-Security and Prudent Principles}\label{sec:property}

In this section, we give a security condition to control the
legitimate release of confidential information on the
where-dimension of security goals. It considers both the code
locality where the release occurs and the level locality to which
security domain the release is legal. Let $\twoheadrightarrow$
represent a (possible empty) sequence of declassification-free
transitions. A trace of computations is separated to the
declassifications labeled with $\rightarrow_d$ and
declassification-free computation sequences. The
\emph{where-security} is formally specified as below.
\begin{definition}[Where-Security]
Program $P$ satisfies \emph{where-security} iff
$\forall\ell\in\mathcal{D}$, we have\\
{\small
$\forall\mathcal{I},\mu,\mathcal{I}',\mu'. \exists n\ge 0:$\\
$\left(
\begin{array}{l}
\forall\mathcal{O}_{n+1},\mu_{n+1}:(\mu,\mathcal{I},\mathcal{O},p,q,P)\lbrack\twoheadrightarrow(\mu_{k_s},\mathcal{I},\mathcal{O}_k,p_k,q_k,\\
x_k:=\mathpzc{declass}(e_k);P_k)\rightarrow_d(\mu_{k_t},\mathcal{I},\mathcal{O}_k,p_k,q_k,P_k)\rbrack_{k=1..n}\\
\twoheadrightarrow(\mu_{n+1},\mathcal{I},\mathcal{O}_{n+1},p_{n+1},q_{n+1},\textbf{skip})\\
\wedge \mathcal{I}\sim_\ell\mathcal{I}' \wedge \mu\sim_\ell\mu'
\end{array}\right)
\Rightarrow$\\
$\left(\begin{array}{l}\exists\mathcal{O}_{n+1}',\mu_{n+1}':
(\mu',\mathcal{I}',\mathcal{O}',p',q',P)\lbrack\twoheadrightarrow(\mu_{k_s}',\mathcal{I}',\mathcal{O}_k',p_k',q_k',\\
x_k':=\mathpzc{declass}(e_k');P_k')\rightarrow_d(\mu_{k_t}',\mathcal{I}',\mathcal{O}_k',p_k',q_k',P_k')\rbrack_{k=1..n}\\
\twoheadrightarrow(\mu_{n+1}',\mathcal{I}',\mathcal{O}_{n+1}',p_{n+1},q_{n+1}',\textbf{skip})\\
\wedge \bigwedge_{k=1..n}(\mu_{k_s}\sim_\ell\mu_{k_s}'\wedge
\mu_{k_s}(e_k)=\mu_{k_s}'(e_k')\Rightarrow
\mu_{k_t}\sim_\ell\mu_{k_t}')\\
\wedge
\left(\begin{array}{l}\bigwedge_{k=1..n}(\mu_{k_s}(e_k)=\mu_{k_s}'(e_k'))\Rightarrow\\
\mu_{n+1}\sim_\ell\mu_{n+1}'\wedge\mathcal{O}_{n+1}\sim_\ell\mathcal{O}_{n+1}'\end{array}\right)\end{array}\right)$}
\end{definition}
Intuitively speaking, when the indistinguishable relation on the
final states is violated, the contrapositive implies that it is
caused by the variation of declassified expressions. This variation
is indicated valid by the premise our property. If the leakage of
confidential information is caused by a computation other than the
primitive $\mathpzc{declass}$, it will be captured because without
constraining the equality of released expression, the final
indistinguishability cannot hold. Our where-security property is
more relaxed than WHERE
\cite{DBLP:conf/esop/MantelR07,DBLP:conf/ifip1-7/LuxM08} which uses
strong-bisimulation and requires each declassification-free
computation step meets the baseline noninterference. We can use
explicit final output of public variables to adapt the judgement of
$\mu_{n+1}\sim_\ell\mu_{n+1}'$ to the judgement of
$\mathcal{O}_{n+1}\sim_\ell\mathcal{O}_{n+1}'$.

Sabelfeld and Sands\cite{DBLP:journals/jcs/SabelfeldS09} clarify
four basic prudent principles for declassification policies as
sanity checks for the new definition: \emph{semantic consistency},
\emph{conservativity}, \emph{monotonicity of release}, and
\emph{non-occlusion}. Our where-security property can be proved to
comply with the first three principles. Let $P[C]$ represent a
program contains command $C$. $P[C'/C]$ substitutes each occurrence
of $C$ in $P$ with $C'$. The principles with respect to the
where-security are defined as follows.

\begin{lemma}[Semantic Consistency]\label{lemma:consist}
Suppose $C$ and $C'$ are declassification-free commands and
semantically equivalent on the same domain of configuration. If
program $P[C]$ is where-secure, the $P[C'/C]$ is where-secure.
\end{lemma}
\begin{lemma}[Conservativity]\label{lemma:conservativity}
If program $P$ is where-secure and $P$ contains no declassification,
then $P$ satisfies noninterference property.
\end{lemma}
\begin{lemma}[Monotonicity of Release]\label{lemma:monotonicity}
If program $P[x:=e]$ is where-secure, then
$P[x:=\mathpzc{declass}(e)/x:=e]$ is where-secure.
\end{lemma}
\begin{corollary}\label{coro}
The where-security satisfies semantic consistency, conservativity,
and monotonicity of release.
\end{corollary}
This corollary indicates that the where-security complies with the
three prudent principles given by the above lemmas. The proofs of
the lemmas are presented in \cite{thispaper_tr}. The non-occlusion
principle cannot be formally proved since a proof would require a
characterization of secure information flow which is what we want to
check against the prudent principles.

\section{Enforcement}\label{sec:enforcement}

In this section, we provide a new enforcement for the where-security
based on reachability analysis of symbolic pushdown
system\cite{schwoonphd}. A pushdown system is a stack-based state
transition system whose stack contained in each state can be
unbounded. It is a natural model of sequential program with
procedures. Symbolic pushdown system is a compact representation of
pushdown system encoding the variables and computations
symbolically.
\begin{definition}[Symbolic Pushdown System, SPDS]
Symbolic Pushdown System is a triple
$\mathcal{P}=(\mathcal{G},\Gamma\times\mathcal{L},\Delta)$.
$\mathcal{G}$ and $\mathcal{L}$ are respectively the domain of
global variables and local variables. $\Gamma$ is the stack
alphabet. $\Delta$ is the set of symbolic pushdown rules
$\{\langle\gamma\rangle\hookrightarrow\langle\gamma_1\cdots\gamma_n\rangle
(\mathcal{R})\mid\gamma,\gamma_1,\cdots,\gamma_n\in\Gamma\wedge
\mathcal{R}\subseteq(\mathcal{G}\times\mathcal{L})\times(\mathcal{G}\times\mathcal{L}^n)\wedge
n\leq 2\}$.
\end{definition}
The stack symbols denote the flow graph nodes of program. The
relation $\mathcal{R}$ specifies the variation of abstract variables
before and after a single step of symbolic execution directed by the
pushdown rules. The operations on $\mathcal{R}$ are compactedly
implemented with binary decision diagrams
(BDDs)\cite{DBLP:journals/tc/Bryant86} in Moped\cite{moped} which we
use as the back-end verification engine.

\begin{table*}[!t]
\renewcommand{\arraystretch}{1.3}
\caption{PDS Rules for Model Construction}\label{table:pds_rules}
\centering
\begin{tabular}{|c|c|}\hline
IR$_H$ &
$\langle\gamma_j\rangle\hookrightarrow\langle\gamma_k\rangle\
(x'=\bot)\wedge
rt(\mu\setminus\{x\},\mathcal{I}^\ell,\mathcal{O}^\ell,p^\ell,q^\ell,\cdots)$\\\hline

IR$_L$ &
$\langle\gamma_j\rangle\hookrightarrow\langle\gamma_k\rangle\
(x'=\mathcal{I}_i[p_i])\wedge(p_i'=p_i+1)\wedge
rt(\mu\setminus\{x\},\mathcal{I}^\ell,\mathcal{O}^\ell,p^\ell\setminus\{p_i\},q^\ell,\cdots)$\\\hline

OR$_H$ &
$\langle\gamma_j\rangle\hookrightarrow\langle\gamma_k\rangle\
rt(\mu,\mathcal{I}^\ell,\mathcal{O}^\ell,p^\ell,q^\ell,\cdots)$
\\\hline

OR$_L$ & $\langle\gamma_j\rangle\hookrightarrow\langle
\textit{output}_{\text{entry}} \gamma_k\rangle\
(\textit{tmp}'=e)\wedge
rt(\mu,\mathcal{I}^\ell,\mathcal{O}^\ell,p^\ell,q^\ell,\cdots)\wedge
rt_2(\cdots)$ \\

& $\langle
\textit{output}_{\text{exit}}\rangle\hookrightarrow\langle\epsilon\rangle\
rt(\mu,\mathcal{I}^\ell,\mathcal{O}^\ell,p^\ell,q^\ell,\cdots)$\\\hline

DR & $ \langle\gamma_j\rangle\hookrightarrow\langle
\mathpzc{declass}_{\text{entry}}^{\gamma_j}\rangle\
(\textit{tmp}'=e)\wedge
rt(\mu,\mathcal{I}^\ell,\mathcal{O}^\ell,p^\ell,q^\ell,\cdots)$\\

& $\langle
\mathpzc{declass}_{\text{exit}}^{\gamma_j}\rangle\hookrightarrow\langle\gamma_k\rangle\
rt(\mu,\mathcal{I}^\ell,\mathcal{O}^\ell,p^\ell,q^\ell,\cdots)$\\\hline
\end{tabular}
\end{table*}

\begin{table*}[!h]
\renewcommand{\arraystretch}{1.3}
\caption{Stuffer PDS Rules for Model
Transformation}\label{table:pds_rules2} \centering
\begin{tabular}{|c|c|}\hline
RST &
$\langle\gamma_j\rangle\hookrightarrow\langle\xi(\gamma_0)\rangle\
(\forall p_i\in p^\ell.p_i'=0)\wedge (\forall q_i\in q^\ell.q_i'=0)
\wedge
rt(\mu,\xi(\mu),\mathcal{I}^\ell,\mathcal{O}^\ell,\cdots)$\\\hline

OS$_i$ & $\langle
\textit{output}_{\text{entry}}\rangle\hookrightarrow\langle
\textit{output}_{\text{exit}}\rangle\
(\mathcal{O}_i'[q_i]=\textit{tmp})\wedge(q_i'=q_i+1)\wedge
rt(\mu,\xi(\mu),\mathcal{I}^\ell,\mathcal{O}^\ell\setminus\{\mathcal{O}_i[q_i]\},p^\ell,q^\ell\setminus\{q_i\},\cdots)$\\\hline

OM$_i$ & $
\langle\xi(\textit{output}_{\text{entry}})\rangle\hookrightarrow\langle\textit{error}\rangle\
(\mathcal{O}_i[q_i]\neq \textit{tmp})\wedge rt(\cdots)$\\

& $\langle
\xi(\textit{output}_{\text{entry}})\rangle\hookrightarrow\langle
\xi(\textit{output}_{\text{exit}})\rangle\
(\mathcal{O}_i[q_i]=\textit{tmp})\wedge (q_i'=q_i+1)\wedge
rt(\mu,\xi(\mu),\mathcal{I}^\ell,\mathcal{O}^\ell,p^\ell,q^\ell\setminus\{q_i\},\cdots)$\\\hline

DS$_{\gamma_j}$ & $\langle
\mathpzc{declass}_{\text{entry}}^{\gamma_j}\rangle\hookrightarrow\langle
\mathpzc{declass}_{\text{exit}}^{\gamma_j}\rangle\
(\mathcal{D}'[\rho(\gamma_j)]=\textit{tmp})\wedge
(x'=\textit{tmp})\wedge
rt(\mathcal{D}\setminus\{\mathcal{D}[\rho(\gamma_j)]\},\mu\setminus\{x\},\xi(\mu),\cdots)$\\\hline

DM$_{\gamma_j}$ & $
\langle\xi(\mathpzc{declass}_{\text{entry}}^{\gamma_j})\rangle\hookrightarrow\langle\textit{idle}\rangle\
(\mathcal{D}[\rho(\gamma_j)]\neq \textit{tmp})\wedge
rt(\mathcal{D},\cdots)$\\

&
$\langle\xi(\mathpzc{declass}_{\text{entry}}^{\gamma_j})\rangle\hookrightarrow\langle\xi(\mathpzc{declass}_{\text{exit}}^{\gamma_j})\rangle\
(\mathcal{D}[\rho(\gamma_j)]=\textit{tmp})\wedge
(\xi(x)'=\textit{tmp})\wedge
rt(\mathcal{D},\mu,\xi(\mu)\setminus\{\xi(x)\},\cdots)$\\\hline
\end{tabular}
\end{table*}

The model construction of commands other than I/O operations is
similar to the one in our previous work\cite{suncong1}. In the pushdown
system, the public channels are represented by global linear lists.
In another word, for a security domain $\ell\in\mathcal{D}$, we only
model the channels in $\mathcal{I}^\ell$ and $\mathcal{O}^\ell$. Take a
input command for example, if the source channel is $\mathcal{I}_i$,
the pushdown rule has a form of IR$_H$ for
$\sigma(\mathcal{I}_i)\succ\ell$ and IR$_L$ for
$\sigma(\mathcal{I}_i)\preceq\ell$ in Table \ref{table:pds_rules},
where $\bot$ denotes an indefinite value. On the other hand, if the
target channel of output is $\mathcal{O}_i$, the pushdown rule has a
form of OR$_H$ for $\sigma(\mathcal{O}_i)\succ\ell$ and OR$_L$ for
$\sigma(\mathcal{O}_i)\preceq\ell$ in Table \ref{table:pds_rules}.
OR$_H$ is just like a transition of \textbf{skip} since the
confidential outputs do not influence the public part of subsequent
states. The variable \textit{tmp} stores the value of expression to
be outputted or declassified. $rt$ means retainment on value of
global variables and on value of local variables in
$\langle\gamma_j\rangle\hookrightarrow\langle\gamma_k\rangle$.
$rt_2$ for a rule
$\langle\gamma_j\rangle\hookrightarrow\langle\textsf{f}_{\text{entry}}\gamma_k\rangle$
denotes retainment on value of local variables of the caller of
procedure \textsf{f}. The declassifications are modeled with DR in
Table \ref{table:pds_rules}. The bodies of outputs to different
public channel and the bodies of declassifications are vacuous.
These absent parts of model will be filled by the self-composition.
This treatment is decided by the store-match pattern which we
develop to avoid the duplication of public channels and to guide the
instrumented computation to fulfil the premise of where-security
property.

We follow the principle of reachability analysis for noninterference
which we proposed in \cite{suncong2}. The self-composition is
evolved into three phases: basic self-composition, auxiliary initial
interleaving assignments, and illegal-flow state construction. For
simplicity, we use the \emph{compact
self-composition}\cite{suncong1} as basic self-composition. To avoid
duplicating the input channels, we reuse the content of public input
channels by resetting the indices of $p^\ell$ to 0 at the beginning
of the pairing part of model, see RST in Table
\ref{table:pds_rules2}. This treatment is safe because from the
semantics we know that no computation actually modifies the content
of input channels. In order to avoid duplicating the output
channels, we propose a store-match pattern of output actions. This
is to stuff the model after basic self-composition with the pushdown
rules OS and OM in Table \ref{table:pds_rules2} parameterized with
the channel identifier $i$. The OM rules show that when the output
to channel $\mathcal{O}_i$ is computed in the second run, it is
compared with the corresponding output stored during the first run.
If they are not equal, the symbolic execution is directed to the
illegal-flow state \textit{error}.


\renewcommand{\algorithmiccomment}[1]{// #1}
\algsetup{linenosize=\scriptsize, linenodelimiter=. }
\algsetup{indent=0.45em}

\begin{algorithm}[!t]\scriptsize
\caption{Model Transformation}\label{algo}
\begin{algorithmic}[1]
\STATE
$\Delta'\leftarrow\{\langle\gamma_{\text{init}}\rangle\hookrightarrow\langle
\textit{startConf}(\mathcal{P})\rangle\ (\forall
x\in\mbox{dom}(\mu^\ell).\xi(x)'=x)\wedge
rt(\mu,\mathcal{I}^\ell,\mathcal{O}^\ell,p^\ell,q^\ell)\}$

\FORALL{$r\in\Delta\wedge r\neq \textit{LastTrans}(\mathcal{P})$}

\STATE $\Delta'\leftarrow\Delta'\cup\{r.expr\ r.\mathcal{R}\wedge
rt(\xi(\mu))\}$

\ENDFOR

\FORALL{$r\in\Delta$}

\IF{$r.expr=\langle\gamma_j\rangle\hookrightarrow\langle\gamma_s\gamma_k\rangle$}

\STATE
$\Delta'\leftarrow\Delta'\cup\{\langle\xi(\gamma_j)\rangle\hookrightarrow\langle\xi(\gamma_s)\xi(\gamma_k)\rangle\
r.\mathcal{R}_{x\in\textit{Var}}^{\xi(x)}\wedge rt(\mu)\}$

\ELSIF{$r.expr=\langle\gamma_j\rangle\hookrightarrow\langle
\mathpzc{declass}_{\text{entry}}^{\gamma_j}\rangle$}

\STATE
$\Delta'\leftarrow\Delta'\cup\{\langle\xi(\gamma_j)\rangle\hookrightarrow\langle\xi(\mathpzc{declass}_{\text{entry}}^{\gamma_j})\rangle\
r.\mathcal{R}_{x\in\textit{Var}}^{\xi(x)}\wedge rt(\mu)\}\cup
\text{DS}_{\gamma_j}\cup \text{DM}_{\gamma_j}$

\ELSIF{$r.expr=\langle\gamma_j\rangle\hookrightarrow\langle\gamma_k\rangle$}

\STATE
$\Delta'\leftarrow\Delta'\cup\{\langle\xi(\gamma_j)\rangle\hookrightarrow\langle\xi(\gamma_k)\rangle\
r.\mathcal{R}_{x\in\textit{Var}}^{\xi(x)}\wedge rt(\mu)\}$

\ELSIF{$r\neq \textit{LastTrans}(\mathcal{P})$}

\STATE
$\Delta'\leftarrow\Delta'\cup\{\langle\xi(\gamma_j)\rangle\hookrightarrow\langle\epsilon\rangle\
r.\mathcal{R}_{x\in\textit{Var}}^{\xi(x)}\wedge rt(\mu)\}$

\ELSE

\STATE
$\Delta'\leftarrow\Delta'\cup\{\langle\xi(\gamma_j)\rangle\hookrightarrow\langle\xi(\gamma_j)\rangle\
r.\mathcal{R}_{x\in\textit{Var}}^{\xi(x)}\wedge
rt(\mu)\}\cup\{\langle\gamma_j\rangle\hookrightarrow\langle\xi(\textit{startConf}(\mathcal{P}))\rangle\
\text{RST}\}$

\ENDIF

\ENDFOR

\STATE
$\Delta'\leftarrow\Delta'\cup\bigcup_{\mathcal{O}_i\in\mathcal{O}^\ell}(\text{OS}_i\cup
\text{OM}_i)$
\end{algorithmic}
\end{algorithm}

Compared with the noninterference property, the premise of
where-security contains equality relations on the declassified
expressions, therefore we need some structure to instrument the
semantics of abstract model to make sure the computation can proceed
only when the equality relations are satisfied. We define another
global linear list $\mathcal{D}$. Suppose there are $m$
declassifications respectively at code location $\gamma_{d_i}(0\leq
i<m)$ and a function $\rho$ mapping $\gamma_{d_i}$ to $i$. We give
another pattern of store-match that stores the value of expression
declassified at $\gamma_{d_i}$ to the site
$\mathcal{D}[\rho(\gamma_{d_i})]$, see DS in Table
\ref{table:pds_rules2}. The corresponding match operation has a form
of DM in Table \ref{table:pds_rules2}. Note that $\xi$ is the rename
function on the stack symbols to generate new flow graph nodes as
well as on the variables to generate the companion variables for the
pairing part of model. The state \textit{idle} has only itself as
the next state. From the reachability of \textit{error} we can
ensure the violation of where-security without considering the
equality relations on the subsequent outputs. The self-composition
algorithm is given in Algorithm~\ref{algo}. The \textit{LastTrans}
returns the pushdown rule with respect to the last return command of
program. The first rule added to $\Delta'$ denotes the initial
interleaving assignments from public variables to their companion
variables. $r.\mathcal{R}_{x\in\textit{Var}}^{\xi(x)}$ means a
relation substituting each variable in \textit{Var} with the renamed
companion variable.

%




\begin{theorem}[Correctness]\label{theorem:correctness}
Let $SC(\mathcal{P}^\ell)$ be the pushdown system w.r.t. security
domain $\ell$ generated by our self-composition on the model of
program $P$. If $\forall\ell\in\mathcal{D}$, the state \emph{error}
of $SC(\mathcal{P}^\ell)$ is unreachable from any initial state, we
have $P$ satisfies the where-security.
\end{theorem}
(The proof is sketched in the technical report\cite{thispaper_tr})

\section{Evaluation}\label{sec:evaluation}

We implement Algorithm~\ref{algo} as part of the parser of
Remopla\cite{remoplaintro} and use Moped as the black-box back-end
engine for the reachability analysis. Here we use experiments to
evaluate:
\begin{compactenum}[1.]
\item whether the property defined by where-security is similar to the existing properties on the
where-dimension,
e.g.\cite{DBLP:conf/esop/MantelR07,DBLP:conf/sp/AskarovS07}, and
what is the real difference between these properties.

\item the preciseness of the mechanism compared with the type
systems on enforcing the respective security properties.

\item whether the store-match pattern can really reduce the state space as
well as the cost of verification.
\end{compactenum}
The experiments are performed on a laptop with 1.66GHz Intel Core 2
CPU, 1GB RAM and Linux kernel 2.6.27-14-generic. The test cases are
chosen from related works, see Table \ref{table:comparison}.

Firstly, we illustrate that where-security is more relaxed than
WHERE\cite{DBLP:conf/esop/MantelR07,DBLP:conf/ifip1-7/LuxM08} and
gradual release\cite{DBLP:conf/sp/AskarovS07}. Lux and
Mantel\cite{DBLP:conf/ifip1-7/LuxM08} have proposed another two
prudent principles: \emph{noninterference up-to} and
\emph{persistence}. Compared with the four basic principles, the two
principles are not generally used for policies on different
dimensions. The conformances of the properties with these principles
are given in Table \ref{table:diff}. Similar to the gradual release,
the program P1 in Table \ref{table:comparison} is secure (denoted by
\checkmark) w.r.t. where-security. This indicates the two properties
do not comply with persistence since the reachable command $l:=h$ is
obviously not secure. On the contrary, WHERE rejects this program.
Our where-security does not comply with noninterference up-to
because the definition deduces relations on final states but not on
the states before $\mathpzc{declass}$ primitives. A typical example
is P0. It is where-secure but judged insecure by WHERE and gradual
release. Although different on these special cases, the
where-security can characterize a similar property to WHERE and
gradual release for the most cases in Table \ref{table:comparison},
see the column \emph{WHERE}, \emph{GR} and \emph{where}.



\begin{table}[!t]
\renewcommand{\arraystretch}{1.3}
\caption{Difference between Properties}\label{table:diff} \centering
\begin{tabular}{|l|c|c|c|}\hline
 & \emph{WHERE} & \emph{gradual release} & \emph{where}\\\hline

noninterference up-to & \checkmark & \checkmark & $\times$\\\hline

persistence & \checkmark & $\times$ & $\times$\\\hline
\end{tabular}
\end{table}


\begin{table*}[!t]
\renewcommand{\arraystretch}{1.3}
\caption{Property and Enforcement Comparison with WHERE and Gradual
Release}\label{table:comparison} \centering
\begin{minipage}[lt]{11.2cm}
\begin{tabular}[!l]{|c|c|c|c|c|c|c|c|c|c|}\hline
Case & From & \emph{WHERE} & $\tau_1$ & \emph{GR} & $\tau_2$ &
\emph{where} & RA & T(ms) & N$_{min}$\\\hline

Ex2 & Example 2,\cite{DBLP:conf/aplas/MantelS04} & $\times$ &
$\times$ & $\times$ & $\times$ & $\times$ & $\times$ & 39.2 &
2\\\hline

RSA & Example 5,\cite{DBLP:conf/aplas/MantelS04} & $\times$ &
$\times$ & $\times$ & $\times$ & $\times$ & $\times$ & 1.09 & 1
\\\hline

C1 & Example 1,\cite{DBLP:conf/esop/MantelR07} & $\times$ & $\times$
& $\times$ & $\times$ & $\times$ & $\times$ & 0.55 & 1 \\\hline

C2 & Example 1,\cite{DBLP:conf/esop/MantelR07} & \checkmark &
\checkmark & \checkmark & \checkmark & \checkmark & \checkmark &
0.59 & --\\\hline

C3 & Example 1,\cite{DBLP:conf/esop/MantelR07} & \checkmark &
\checkmark & \checkmark & \checkmark & \checkmark & \checkmark &
0.49 & --\\\hline

filter & Fig.6,\cite{DBLP:conf/esop/MantelR07} & \checkmark &
\checkmark & \checkmark & \checkmark & \checkmark & \checkmark &
5.47 & --
\\\hline

P0 & Sec.1,\cite{DBLP:conf/pldi/AskarovS07} & $\times$ & $\times$ &
$\times$ & $\times$ & \checkmark & \checkmark & 0.44 & --\\\hline

P1 & Sec.2,\cite{DBLP:conf/sp/AskarovS07} & $\times$ & $\times$ &
\checkmark & $\times$ & \checkmark & \checkmark & 0.53 & --\\\hline

P2 & Sec.3,\cite{DBLP:conf/pldi/AskarovS07} & \checkmark & $\times$
& \checkmark & $\times$ & \checkmark & \checkmark & 0.64 &
--\\\hline

P3 & Sec.2,\cite{DBLP:conf/sp/AskarovS07} & $\times$ & $\times$ &
$\times$ & $\times$ & $\times$ & $\times$ & 3.53 & 1\\\hline

P4 & Sec.4,\cite{DBLP:conf/pldi/AskarovS07} & $\times$ & $\times$ &
$\times$ & $\times$ & $\times$ & $\times$ & 2.03 & 1\\\hline

P5 & Sec.4,\cite{DBLP:conf/pldi/AskarovS07} & $\times$ & $\times$ &
$\times$ & $\times$ & $\times$ & $\times$ & 0.61 & 1\\\hline

P6 & Sec.5,\cite{DBLP:conf/pldi/AskarovS07} & \checkmark & $\times$
& \checkmark & $\times$ & \checkmark & \checkmark & 0.37 &
--\\\hline

P7 & Sec.2,\cite{DBLP:conf/sp/AskarovS07} & \checkmark & $\times$ &
\checkmark & \checkmark & \checkmark & \checkmark & 1.91 &
--\\\hline
\end{tabular}
\end{minipage}
\begin{minipage}[rt]{6cm}
\begin{tabular}{c l}\hline
P0 & $l:=h;l:=\mathpzc{declass}(h);$\\
P1 & $l:=\mathpzc{declass}(h);l:=h;$\\
P2 & $h_1:=h_2;l:=\mathpzc{declass}(h_1);$\\
P3 & $h_1:=h_2;h_2:=0;$\\
   & $l_1:=\mathpzc{declass}(h_2);h_2:=h_1;l_2:=h_2;$\\
P4 & $h_2:=0;$\\
   & $\text{if }h_1\text{ then }l:=\mathpzc{declass}(h_1)$\\
   & $\text{else }l:=\mathpzc{declass}(h_2);$\\
P5 & $l:=0;$\\
   & $\text{if }l\text{ then }l:=\mathpzc{declass}(h)\text{ else skip};$\\
   & $l:=h;$\\
P6 & $h_2:=0;$\\
   & $\text{if }h_1\text{ then }l:=\mathpzc{declass}(h_2)\text{ else }l:=0;$\\
P7 & $l:=\mathpzc{declass}(h!=0);$\\
   & $\text{if }l\text{ then }l_1:=\mathpzc{declass}(h_1)\text{ else skip};$\\\hline
\end{tabular}
\end{minipage}
\end{table*}


Then we evaluate the preciseness of our enforcement mechanism. In
Table \ref{table:comparison} $\tau_1$ is the well-typeness of
program judged by the type system in
Fig.4,\cite{DBLP:conf/esop/MantelR07}. $\tau_2$ is the judgement of
the type system given in Fig.3,\cite{DBLP:conf/sp/AskarovS07}. RA is
the reachability analysis result using our mechanism. \checkmark means the state
\textit{error} is not reachable. The analysis time T is related to
the number of bits of each variable, which we set to 3 and that
means each variable in the model has a range of 0$\sim$2$^3$-1.
Larger number of bits corresponds to the increase on state space of
model and the analysis time. On the other hand, the number of bits
of variable is meaningful also because if it is too small for the
model of insecure program, the illegal path cannot be caught. This
causes a false-positive which can be avoided by setting the number
of bits of variable sufficiently large. We record the minimum number
of bits to avoid false-positive as N$_{min}$. The analysis might be
time consuming when N$_{min}$ is large. For secure program, the
illegal-flow state will be unreachable for any number of bits
therefore N$_{min}$ is not recorded. The program \emph{filter} in
Table \ref{table:comparison} has a more complex policy. From the
escape hatch information we have
$\emph{reader}\preceq\emph{network}$. The model is constructed and
transformed on respective security domains. On each security domain
different public variables are modeled outputted in the end and
state \textit{error} of transformed model is unreachable. Our
enforcement is more precise compared with the type systems that
reject some secure programs (P2,P6,P7 for WHERE and P1,P2,P6 for
gradual release).




\begin{figure}[!t]
\centering
\includegraphics[scale=.38,bb=80 295 580 530]{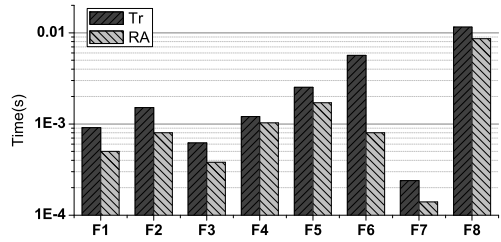}
\caption{Cost Reduction with Store-Match Pattern}\label{fig:sto_mat}
\end{figure}

Finally, we evaluate the reduction on the cost of
verification provided by the store-match pattern. We compare our
mechanism with a model transformation, i.e. \textit{Tr} in
Fig.\ref{fig:sto_mat}, which duplicates the public output channels
and constructs the illegal-flow state following the pairing part of
model. The test cases containing I/Os are from
Fig.4,\cite{DBLP:journals/ijisec/FrancescoM07}, and named
F$_1$$\sim$F$_8$ in Fig.\ref{fig:sto_mat}. These experiments show
that the store-match pattern can give an overall 41.4\% reduction on
the cost of verification. The number of bits of variable is set to 3
as well.

\section{Conclusion}\label{sec:conclusion}

We propose a security property on the where-dimension of
declassification. The property is proved complying with the three
classical prudent principles. We also give a precise enforcement
based on the reachability analysis of pushdown system derived by a
variant of self-composition. To immigrate our approach to the
properties on other dimensions of declassification, e.g. the
\emph{delimited release}\cite{DBLP:conf/isss2/SabelfeldM03} on the
what-dimension, the key point is to focus on the
indistinguishability of declassified expressions on the pair of
initial states. The study on the enforcement of properties on the
other dimensions is left to our future work.

\section*{Acknowledgment}

We thank Alexander Lux for providing the valuable proofs and
explanations of the theorems in their work. We also thank Ennan Zhai for helpful comments
and the anonymous reviewers for useful feedback.
This research is partially supported by the National Natural Science
Foundation of China under Grant No.60773163, No.60911140102,
The National Key Technology R\&D Program in the 11th five-year Period under Grant No.2008BAH33B01,
as well as the PKU Project PKU-PY2010-005.




\bibliographystyle{IEEEtran}
\bibliography{IEEEabrv,mybib}

%
%


\newpage

\section*{Appendix}

\begin{proof}[Proof of Lemma \ref{lemma:consist}]
Suppose any trace of the program $P[C'/C]$ is in a form of\\
$(\mu,\mathcal{I},\mathcal{O},p,q,P[C'/C])\rightarrow^*
(\mu_j,\mathcal{I},\mathcal{O}_j,p_j,q_j,C';P_j)\twoheadrightarrow
(\mu_k,\mathcal{I},\mathcal{O}_k,p_k,q_k,P_j)\rightarrow^*
(\mu_f,\mathcal{I},\mathcal{O}_f,p_f,q_f,\textbf{skip})$.\\
Because $C$ and $C'$ are semantically equivalent, we also have
$(\mu_j,\mathcal{I},\mathcal{O}_j,p_j,q_j,C;P_j)\twoheadrightarrow
(\mu_k,\mathcal{I},\mathcal{O}_k,p_k,q_k,P_j)$. Moreover, since $C$
and $C'$ are declassification-free, the substitution will not
influence the conjunction of equivalence on declassified expressions
in $P[C'/C]$. Therefore the indistinguishability on the final
configurations, that is
$\mu_{n+1}\sim_\ell\mu_{n+1}'\wedge\mathcal{O}_{n+1}\sim_\ell\mathcal{O}_{n+1}'$,
holds before and after the substitution.
\end{proof}
\begin{proof}[Proof of Lemma \ref{lemma:conservativity}]From the
operational semantics we can see $\rightarrow_d$ can only occurs
when a $\mathpzc{declass}$ command is executed and the declassified
expression contains some information with a security domain higher
than the security domain of $x$. $P$ has no declassification implies
that in any trace of computation of $P$ there is no $\rightarrow_d$.
The where-security of $P$ degenerates to have $n=0$. Therefore the
where-security becomes noninterference according to the definition
and $\mu_f\equiv\mu_1,\mathcal{O}_f\equiv\mathcal{O}_1$.
\end{proof}
\begin{proof}[Proof of Lemma \ref{lemma:monotonicity}] There are
actually two cases on whether the substitution introduces a real
declassification.
\begin{compactenum}[1.]
\item If $\sigma(e)\preceq\sigma(x)$, the computation of
$x:=\mathpzc{declass}(e)$ is identical to the ordinary assignment
$x:=e$ and $\rightarrow$ is not labeled as $\rightarrow_d$. The
where-security of $P[x:=\mathpzc{declass}(e)/x:=e]$ does not change
compared with the where-security of $P[x:=e]$.

\item Suppose we have $\sigma(x)\prec \sigma(e)$. The computation of $x:=e$ in the two correlative runs of $P[x:=e]$ are
like $(\mu,\mathcal{I},\mathcal{O},p,q,P[x:=e])\rightarrow^*
(\mu_j,\mathcal{I},\mathcal{O}_j,p_j,q_j,x:=e;P_j)\rightarrow
(\mu_j[x\mapsto\mu_j(e)],\mathcal{I},\mathcal{O}_j,p_j,q_j,P_j)\rightarrow^*
(\mu_{n+1},\mathcal{I},\mathcal{O}_{n+1},p_{n+1},q_{n+1},\textbf{skip})$
and $(\mu',\mathcal{I}',\mathcal{O}',p',q',P[x:=e])\rightarrow^*
(\mu_j',\mathcal{I}',\mathcal{O}_j',p_j',q_j',x:=e;P_j)\rightarrow
(\mu_j'[x\mapsto\mu_j'(e)],\mathcal{I}',\mathcal{O}_j',p_j',q_j',P_j')\rightarrow^*
(\mu_{n+1}',\mathcal{I}',\mathcal{O}_{n+1}',p_{n+1}',q_{n+1}',\textbf{skip})$.
From the premise of where-security of
$P[x:=\mathpzc{declass}(e)/x:=e]$ we have
$\bigwedge_{k=1..n}(\mu_{k_s}\sim_\ell\mu_{k_s}'\wedge
\mu_{k_s}(e_k)=\mu_{k_s}'(e_k))$. That implies
$\bigwedge_{k=1..n,k\neq j}(\mu_{k_s}\sim_\ell\mu_{k_s}'\wedge
\mu_{k_s}(e_k)=\mu_{k_s}'(e_k))$ and because $P[x:=e]$ is
where-secure, we have $\bigwedge_{k=1..n,k\neq
j}(\mu_{k_t}\sim_\ell\mu_{k_t}')$. Because
$\mu_j\sim_\ell\mu_j'\wedge\mu_j(e)=\mu_j'(e)$, according to the
semantics, we have $\mu_j[x\mapsto
\mu_j(e)]\sim_\ell\mu_j'[x\mapsto\mu_j'(e)]$, that is
$\mu_{j_t}\sim_\ell\mu_{j_t}'$ for $P[x:=\mathpzc{declass}(e)/x:=e]$
and therefore $\bigwedge_{k=1..n}(\mu_{k_t}\sim_\ell\mu_{k_t}')$. On
the other hand, since the substitution does not change the semantics
of program, restricting the premise $\bigwedge_{k=1..n,k\neq
j}(\mu_{k_s}(e_k)=\mu_{k_s}'(e_k'))$ with a conjunction to
$\mu_j(e)=\mu_j'(e)$ will not influence the consequence that
$\mu_{n+1}\sim_\ell\mu_{n+1}'\wedge\mathcal{O}_{n+1}\sim_\ell\mathcal{O}_{n+1}'$.
The where-security of $P[x:=\mathpzc{declass}(e)/x:=e]$ is proved.
\end{compactenum}
\end{proof}

\begin{proof}[Proof of Theorem~\ref{theorem:correctness}]
Suppose program $P$ violates the where-security property, that means
\begin{equation*}
\exists k_0.\mu_{k_0,s}\sim_\ell\mu_{k_0,s}'\wedge
\mu_{k_0,s}(e_k)=\mu_{k_0,s}'(e_k')\wedge\neg
(\mu_{k_0,t}\sim_\ell\mu_{k_0,t}')
\end{equation*}
or
\begin{equation*}
\bigwedge_{k=1..n}(\mu_{k_s}(e_k)=\mu_{k_s}'(e_k'))\wedge\neg(\mathcal{O}_{n+1}\sim_\ell\mathcal{O}_{n+1}')
\end{equation*}
Here the $\mu_{n+1}\sim_\ell\mu_{n+1}'$ has been adapted to
$\mathcal{O}_{n+1}\sim_\ell\mathcal{O}_{n+1}'$ by modeling final
outputs of public variables. If the first relation is satisfied, we
have in $x_k:=\mathpzc{declass}(e_k)$ and
$x_k':=\mathpzc{declass}(e_k')$ the variable $x_k$ and $x_k'$ are
different variables. Therefore the respective pushdown rules must
have different $\gamma_j$ as the label for the stack symbol
$\mathpzc{declass}_{\text{entry}}^{\gamma_j}$, which we suppose to
be $\gamma_k, \gamma_{k'}$ and $\gamma_k\neq\gamma_{k'}$. From the
DS$_{\gamma_k}$ and DM$_{\gamma_{k'}}$ we have
$\mathcal{D}[\rho(\gamma_{k'})]=e_k'$. The value in
$\mathcal{D}[\rho(\gamma_{k'})]$ is irrelevant to $e_k$ and $x_k$ in
the second run is not restricted by DM$_{\gamma_{k'}}$. When the
final $x_k$ and $x_k'$ are outputted, the inequality of final $x_k$
of correlative executions makes the state \textit{error} reachable
according to the rule OM$_{\sigma(x_k)}$. If the second relation is
satisfied, $\exists i.q_i\neq q_i'\vee (\exists 0\leq
k_0<q_i.\mathcal{O}_i[k_0]\neq\mathcal{O}_i'[k_0])$. If $q_i\neq
q_i'$, we can suppose $q_i<q_i'$ because the correlative runs are
symmetrical. Then there must be some $e$ of
$\textit{output}(e,\mathcal{O}_i')$ in $P$ that should be compared
with the indefinite value in $\mathcal{O}_{i,n+1}[q_i]$ during the
execution of the second run. Otherwise we have
$\mathcal{O}_{i,n+1}[k_0]\neq\mathcal{O}_{i,n+1}'[k_0]$. Then if
$\mathcal{O}_{i,n+1}'[k_0]$ is generated by
$\textit{output}(e,\mathcal{O}_i)$, the second run is directed by
$\mathcal{O}_i[k_0]\neq e$ according to the rule OM$_i$ and
\textit{error} is reachable. From the contrapositive the theorem is
proved.
\end{proof}




\end{document}